\documentclass{article}

\usepackage[nonatbib, main, final]{neurips_2026} 

\usepackage[utf8]{inputenc} 
\usepackage[T1]{fontenc}    
\usepackage{hyperref}       
\usepackage{url}            
\usepackage{booktabs}       
\usepackage{amsfonts}       
\usepackage{nicefrac}       
\usepackage{microtype}      
\usepackage{xcolor}         

\usepackage{amsmath,graphicx,hyperref}
\usepackage{bm}
\usepackage{amsthm}
\newtheorem{theorem}{Theorem}

\usepackage{xspace}

\usepackage[
    style = ieee,
    doi=false,isbn=false,url=false,eprint=false
]{biblatex}
\DeclareMathOperator{\Tr}{Tr}                                 
\newcommand{\x}{\ensuremath{\mathbf{x}}\xspace}               
\newcommand{\bxS}{\ensuremath{\x_\mathcal{S}}\xspace}         
\newcommand{\ct}{\ensuremath{\mathbf{c}(t)}\xspace}           
\newcommand{\Lcano}{\ensuremath{L_{\text{cano}}}\xspace}

\title{Local Diagnostics of Continuous Normalizing Flow\\for Out-of-Distribution Detection }

\author{%
  Xinwei Cao
    \\
  Department of Electronic Systems\\
  Norwegian University of Science and Technology\\
  Trondheim, Norway \\
  \texttt{xinwei.cao@ntnu.no} \\
  \And
  Mengxuan Lu \\
  School of Future Technology \\
  South China University of Technology \\
  Guangzhou, China \\
  \texttt{ftlumengxuan@mail.scut.edu.cn} \\
  \AND
  Torbjørn Svendsen \\
  Department of Electronic Systems\\
  Norwegian University of Science and Technology\\
  Trondheim, Norway \\
  \texttt{torbjorn.svendsen@ntnu.no} \\
  \And
  Giampiero Salvi \\
  Department of Electronic Systems\\
  Norwegian University of Science and Technology\\
  Trondheim, Norway \\
  \texttt{giampiero.salvi@ntnu.no} \\
}

\begin{document}

\maketitle
  
\begin{abstract}
We address the problem of out-of-distribution (OOD) detection for target observations embedded in a subspace of the high dimensional data space.
Using continuous normalizing flows~(CNFs), we propose a Lagrangian sub-flow~(LSF) framework designed to isolate and estimate the density for the relevant components in the representation and using the remaining components as context.
Through experimentation with models for speech synthesis, we show that CNFs, similarly to other deep generative models (DGMs), are susceptible to the ``likelihood paradox'', where high likelihood is erroneously assigned to OOD samples.
This is attributed to the inductive bias of DGMs that prioritize low-level structural details over high-level semantic coherence.
To mitigate this phenomenon, we propose a number of geometric diagnostic signals based on the velocity field over the sub-flow trajectory.
Based on these signals, we design metrics for the challenging task of zero-shot phoneme-level mispronunciation detection.
Finally, we demonstrate the superiority of these metrics compared to likelihood-based methods on a real-world mispronunciation detection benchmark.
\end{abstract}

\section{Introduction}
\label{sec:intro}
Deep generative models~(DGMs) have demonstrated an exceptional ability to model complex data distributions.
Yet, their applications in out-of-distribution (OOD) detection underperform dedicated discriminative models due to the so called ``likelihood paradox'' effect. 
As identified by \citeauthor{nalisnick19}~\cite{nalisnick19}, likelihood-based models often assign higher density to OOD samples than to in-distribution data.
For example, a model trained on images from CIFAR-10~\cite{Krizhevsky2009_CIFAR_dataset}, may assign higher likelihoods to images from another domain, such as SVHN~\cite{NetzerEtAl2011_SVHN_dataset}.
In complement to it, Ren et al.\cite{ren2019likelihood} have found that standard likelihood scores are often dominated by background statistics rather than the ``semantic" features that actually define a class.
Similarly, Kirichenko et al.~\cite{kirichenko2020why} provided a geometric explanation, suggesting that DGMs such as normalizing flows~(NF) fail because they prioritize low-level pixel statistics over high-level semantics. 

Efforts have been made to mitigate the ``likelihood paradox'' issue.
For example, researchers managed to use likelihood ratios~(\cite{ren2019likelihood}) or input complexity\cite{serra2020input} to compensate the bias towards background noise.
Le et al.\cite{lelan2021geometric} gave a geometric explanation of why NF models fail to detect simple OOD samples and proposed to estimate the local intrinsic dimensions~(LID) and augment the likelihood with additional criteria.

The above mentioned methods treat DGMs as static probability estimators whilst ignoring the rich dynamical information embedded in the training and inference of the model itself.
Xiao et al.\cite{xiao2020likelihood} proposed to use the likelihood-ratio before and after a training step of DGMs to distinguish between In-domain~(ID) and OOD samples.
More recently, Diffusion models, which have dominated the field of content generation on all modalities, have been applied for OOD detection, e.g, MSMA~\cite{jolicoeur2021msma} or DiffGuard~\cite{gao2023diffguard}.
The former approach is based on the trained score-function evaluated at each point on the noise adding process~(the diffusion path). 
The latter runs the data to noise path first forwards and then procedes backwards with an invertible diffusion model. 
OOD detection is performed by analyzing the property of the reconstructed point.
Latent ODE models~\cite{rubanova2019latent} can be used to model the evolution of time series in the latent space, e.g., the temperature over a period of time.
Morningstar et al.~\cite{morningstar2021density} defined a dynamics prior with a latent ODE model to distinguish between normal and abnormal time-series.

Another powerful category of DGMs, namely the continuous normalizing flow (CNF), is however left understudied for OOD detection. 
One of the major concerns of using CNF for OOD detection is the computational complexity since CNF usually requires multiple steps or number of function evaluations~(NFE) to evaluate the likelihood.
However, we believe that the theoretical analogy between CNF and fluid dynamics offers a powerful perspective for extracting diagnostic signals. 
Moving beyond the holistic perspective of the flow, we argue that the true diagnostic power of CNF lies in the localized, fine-grained trajectories within the high-dimensional manifold, or the sub-flow dynamics.
Due to averaging over high-dimensional representations, in real-world applications, OOD signals rarely correspond to very low likelihoods.
By decomposing the global flow into context-specific sub-flows, we can capture the structural incoherence that is otherwise masked by global averaging. 

There are many application domains of sub-flow diagnostics.
In speech generation, for example, a sub-flow can represent the trajectory of a specific phoneme while the global flow represents global information such as environmental noise and speaker's identity;
In computer vision and text-to-image generation, sub-flows correspond to spatial-semantic regions and analyzing these localized flows may allow for the detection of ``hallucinations" where a specific object’s internal geometry contradicts the global scene;
Furthermore, detecting localized OOD patches in medical imaging may be useful to detect a disease, where the rest of the image acts as calibration for the specific patient.
Finally, in bio-informatics and time-series analysis, anomalies in these sub-flows can signal mutations or transient regime shifts under a global trend.

In this work, we investigate the possibilities of using CNF for OOD detection.
More specifically, we focus on the task of variable scale sub-flow OOD detection within the holistic flow. 
We make the following contributions:
\begin{itemize}
    \item We identify the global coupling issue in modern Transformer-based CNFs and propose the Lagrangian sub-flow~(LSF) framework to inspect local properties of the sub-flow while preserving global context awareness from the original flow.
    \item We provide several stochastic tools that efficiently separate local diagnostic signals from the global evolution, for analyzing the sub-flow dynamics.  
    \item We demonstrate the usage of these tools in the domain of phoneme-level pronunciation assessment using a well-known CNF-based speech generative model.
    \item We visualize and discuss the geometric properties of the learned vector field.
\end{itemize}

\section{Background and Related Work}
\label{sec:background}
In this section, we briefly introduce the continuous normalizing flow (CNF) which is essential for understanding the methods.
Consider observations $x$ in a $D$-dimensional data space $\mathbb{R}^D$ distributed according to an unknown data distribution $q_{\text{data}}(x)$.
CNF-based models learn a continuous transformation with the probability path $p_t(x)$ over flow-time $t\in[0,1]$ from a Gaussian noise source $p_0(x) =\mathcal{N}(x; 0,I_D)$ to the approximated data density $p_1(x) \approx q_{\text{data}}(x)$.  
This transformation corresponds to a flow $\phi(x,t): \mathbb{R}^D\times \mathbb{R} \rightarrow \mathbb{R}^D$ which is generated by a time-dependent velocity field $v(x,t): \mathbb{R}^D\times \mathbb{R} \rightarrow \mathbb{R}^D$.
This process can be described by a time-dependent random variable $X_t$.

Each initial value $x_0 \in \mathbb{R}^D \sim X_{t=0}$ corresponds to a specific trajectory~(streamline) $\x(t)= \phi(x_0,t)$\footnote{Throughout this paper, we distinguish between the spatial coordinate $x_t \in \mathbb{R}^D$ for the global Eulerian view of the flow and the time-dependent trajectory $\x(t)$ for the Lagrange view of the flow.}.
The relation between the points along the trajectory and their velocity vectors $v(\x(t),t)$ can be expressed using the following ordinary differential equation~(ODE):
\begin{equation}
    \label{eqa:vf}
    \frac{d \x(t)}{dt} = v(\x(t),t).
\end{equation} 
The CNF parametrizes the velocity vector field using a neural network~(neural ODE) $v_\theta(\x(t),t)$ dependent on the parameters $\theta$.
Given a trained CNF and a numerical ODE-solver, e.g., Euler or RK4, sampling from the data distribution can be approximated by solving the initial value problem~(IVP) of Eq.\ref{eqa:vf} forwards in time:
\begin{equation}
   x_1 =  x_0 + \int_0^1v_{\theta}(\x(t),t)dt,
\end{equation}
where $x_0 \sim p_0(x)$ is the initial value sampled from the Gaussian noise.

Under the assumption of Lipschitz continuity of the neural ODE, the CNF can run backwards in time from data $\x(1)$  to noise $\x(0)$ and the rule of instantaneous change of variables~(ICV)~\cite{odenet} can be applied for estimating the likelihood. If we call $J(t) = \frac{\partial v_\theta(\x(t),t)}{\partial \x(t)}$ the Jacobian of the ODE along the trajectory, the log likelihood becomes:
\begin{equation}
    \label{eqa:loglikelihood}
    \log p_1(x =\x(1))=\log p_0(x=\x(0)) - \int_{0}^{1}\Tr\left(J(t)\right)dt,
\end{equation}
where $\Tr\left(J(t)\right)$ is the trace of the Jacobian.
The trace can be efficiently approximated by the unbiased Hutchinson’s trace estimator~\cite{hut}:
\begin{equation}
    \label{eqa:hut}
    \Tr\left(J(t)\right) = E_{\epsilon}[\epsilon^T J(t) \epsilon],
\end{equation}
where $\epsilon$ is typically a vector of Gaussian noise and the expectation is computed using Monte Carlo approximation with $N$ samples of $\epsilon$ (usually $N<10$). 
Finally, training the CNF can be done by directly maximizing the log-likelihood of the data drawn from  $q_{\text{data}}(x)$~\cite{ffjord}. 

Unlike traditional CNF training which requires back-propagating through an ODE solver, modern CNF models can be trained efficiently using the flow-matching (FM) objective~\cite{FM}, following the trajectory of the optimal transport (OT) path between a sampled pair ($x_0,x_1$), defined as:
\begin{equation}
    \label{eqa:OT-path}
    \x(t) = (1-t)x_0 + tx_1.
\end{equation}
As we can see, the OT path is a simple linear interpolation between the initial point $x_0$ and the data point $x_1$ and it has a constant time derivative $v(\x(t), t)  = x_1 - x_0 $ along the trajectory.
The objective of FM is then to minimize the expected $L_2$ distance between the velocity field $v_\theta$ estimated by the neural network and the one calculated along the OT path:
\begin{equation}
    \mathcal{L}_{\text{FM}}(\theta) = \mathbb{E}_{t, x_0, x_1} \left[ \left\| v_\theta(\x(t), t) - (x_1 - x_0) \right\|^2 \right]
\end{equation}
where $t \sim \mathcal{U}[0, 1]$, $x_0 \sim p_0(x)$, and $x_1 \sim q_{\text{data}}(x)$.
This improves training stability and reduces the dependence on precise numerical integration~\cite{FM}.

\section{Method}
\label{sec:method}
\subsection{The Lagrangian Sub-Flow Framework}
\label{sec:subflow}
CNF-based generative models parameterize a continuous vector field that transports a latent Gaussian distribution $p_0$ toward the target data distribution $p_1$. 
In high-dimensional generation tasks, modern architectures often rely on global interaction mechanisms, e.g., self-attention in Transformers, which cause the feature dimensions to become intricately entangled. 
This dense coupling ensures an expressive representation but simultaneously creates a significant challenge for OOD detection.

As described in Section~\ref{sec:intro}, we are interested in methods that can focus the OOD detection on a subset of components in the representation, while using the other components as context.
To isolate local content from the global context, we propose the Lagrangian sub-flow (LSF) framework. 
This framework allows for a post-hoc diagnostic analysis of the sub-flow while keeping the original evolution of the global flow unchanged. 

\subsubsection{The Coupled Sub-Flow and Inter-Dimensional Density Leakage}
\label{sec:coupled_sub-flow}
Consider an ambient continuous normalizing flow $\phi(x,t): \mathbb{R}^D\times \mathbb{R} \rightarrow\mathbb{R}^D$ governed by the velocity field $v(x,t)$.
To isolate the local properties of this ambient flow, we decompose the ambient space into a target sub-space $\mathcal{S}$ and its complementary space $\mathcal{C}$, such that $\mathbb{R}^D = \mathcal{S} \times \mathcal{C}$.  
We define the probability path for the sub-flow space $\mathcal{S}$ while following a specific trajectory $\ct$ in the complementary space as: 
\begin{equation}
    p_{\ct}(x_\mathcal{S}) = p_t(x_S,  x_C=\ct).
    \label{eq:subflow_pdf}
\end{equation}
In this way, we respect the authentic instantaneous contextual condition $\ct$ arising from the joint evolution of the global flow while focusing on the change of the density over flow-time $t$ within the target sub-space only.
Note that the definition in Eq.~\ref{eq:subflow_pdf} is an unnormalized conditional density, where the normalization factor is intractable to compute. 
This problem can be mitigated by introducing the likelihood-ratio which will be discussed in Sec.\ref{sec:experiments:pronunciation}.
A more critical problem of this definition is the fact that the sub-space density function is influenced by the evolution of the context $\ct$ in the complementary space, which we refer to herein as ``inter-dimensional density leakage''.

The ambient flow $\phi$ represents an autonomous fluid system. 
Thus, its density $p_t(x)$ satisfies the standard continuity equation~(the  weak form of the law of conservation of mass) for any $x\in \mathbb{R}^D$:
\begin{equation}
    \label{eqa:continue}
    \frac{\partial p_t(x)}{\partial t} = - \nabla_x \cdot (p_t(x)\cdot v),
\end{equation}
that is, the partial derivatives of the PDF with respect to $t$ is equal to the negative divergence of $p_t(x)\cdot v$ at the specific flow-time $t$.
As the sub-space is embedded within the ambient space, the density evolution at any point $x_{\mathcal{S}} \in \mathcal{S}$ must satisfy:
\begin{align}
    \frac{\partial p_{\ct}(x_\mathcal{S})}{\partial t} &= - \nabla_x \cdot [p_t(x_S, x_C=\ct)\cdot v] = \\
                                        &= - \nabla_{x_s} \cdot [p_t(x_S,x_C=\ct)\cdot v_\mathcal{S}] - \underbrace{\nabla_{x_c} \cdot [p_t(x_S,x_C=\ct)\cdot v_{\mathcal{C}}]}_{\text{Inter-dimensional density leakage}}. \label{eqa:leakage}
\end{align}
In a coupled system where $v_{\mathcal{C}}\neq 0$, the second term in Eq.\ref{eqa:leakage} violates the local conservation of mass, rendering the sub-flow non-autonomous.
More specifically, by applying the product rule of differential, the leakage term can be further decomposed into:
\begin{equation}
    \nabla_{x_c} \cdot \left[p_t(x_S,s_C=\ct)\cdot v_{\mathcal{C}}\right] = \underbrace{\nabla_{x_c} \cdot \left[p_t(x_S,x_C=\ct)\right]\cdot v_{\mathcal{C}}}_{\text{Inter-dimensional advection}} + \underbrace{p_t(x_S,x_C=\ct)\cdot \nabla_{x_c} \cdot ( v_{\mathcal{C}})}_{\text{Inter-dimensional compression}},
\end{equation}
where the ``inter-dimensional advection'' term represents the probability mass ``drifting'' out of the the sub-space $\mathcal{S}$ due to transverse global motion~(e.g., rotation or shifting),  while the ``inter-dimensional compression" term represents density fluctuations induced by contraction of dimensions in the complementary space.
Consequently, the sub-flow is heavily entangled with the global evolution, shielding the intrinsic local kinematic signals from direct diagnostic observation.

\subsubsection{Restoration of Autonomy with the Sealed Vector Field}
To enable an atomic diagnostic of a sub-flow while preserving the model's global context-awareness, we define a projected diagnostic field $\hat{\phi}$  and its corresponding vector field $\hat{v}(x,t)$ by nullifying the off-subspace velocity components of the original vector field:
\begin{equation}
\hat{v}_{i}(x,t) =
\begin{cases}
v_{i}(x, t) & v_i \in \mathcal{S} \\
0 & v_i \notin \mathcal{S}.
\end{cases}
\end{equation}
This can be achieved by applying a kinematic sealing $\mathbf{K}_{\mathcal{S}}$ matrix multiplication:
\begin{equation}
\hat{v}(x, t) = \mathbf{K}_{\mathcal{S}} \cdot v(x, t),
\label{eq:projected_field}
\end{equation}
where $\mathbf{K}_\mathcal{S} = \text{diag}(\mathbf{m}_\mathcal{S})$ and $\mathbf{m} \in \{0,1\}^D$ be the binary mask for the sub-space $\mathcal{S}$ also known as indicator function.
By enforcing $v_{C}=0$ at any time $t\in[0,1]$, the inter-dimensional leakage in Eq.~\ref{eqa:leakage} vanishes, thereby restoring the local continuity equation:
\begin{equation}
     \frac{\partial p_{\ct}(x_\mathcal{S})}{\partial t} 
                                        = - \nabla_{x_s} \cdot \left(p_{\ct}(x_\mathcal{S})\cdot v_\mathcal{S}\right).
    \label{eqa:restored}
\end{equation}
This projection acts as kinematic sealing, creating a ``stiff tube'' in the flow that disentangles the sub-flow from density exchange caused by context advection or compression.
Consequently, the mass is preserved within the sub-flow. 

It is critical to emphasize that the kinematic sealing $\mathbf{K_{\mathcal{S}}}$ is implemented as a post-hoc diagnostic operator applied to the pre-generated trajectory according to its original, coupled dynamics.
While the underlying model $v_\theta(x,t)$ continues to perceive the global context $\ct$ which jointly evolves for coherent generation, the virtual vector field $\hat{v}(x,t)$ is used exclusively for diagnostic. 
This formulation allows us to extract the sub-segment as a locally autonomous kinematic system without redefining a separate dynamical evolution system.

\subsubsection{The Lagrangian Sub-Flow via Post-Hoc Kinematic Sealing}
While the restored continuity equation~(Eq.\ref{eqa:restored}) provides an Eulerian view of the sealed sub-system, we seek to quantify the density evolution from a Lagrangian perspective~(per-sample trajectory) to enable sample-wise diagnostics. 
Due to the post-hoc kinematic sealing mechanism, the sample sub-path remains strictly consistent with the projected trajectory of the original flow.
To this end, we define the sub-path flow map $\x_{\mathcal{S}}(t)$ which tracks the trajectory of a specific initial state $\x_{\mathcal{S}}(0) \in \mathcal{S}$ evolving under the sealed vector field.
Similarly to Eq.\ref{eqa:vf}, the sub-trajectory satisfies the ODE: 
\begin{equation}    
\frac{d \x_{\mathcal{S}}(t)}{dt} = \hat{v}_{\mathcal{S}}(\x(t),t),
\end{equation}
where $\x(t)=[\x_{\mathcal{S}}(t), \ct]$ is a state vector point on the ambient trajectory.

\begin{theorem}{[Instantaneous change of variables for sub-flows]}
Given the restored Lagrangian autonomy of the sealed sub-flow, the total derivative of the log-density along the sub-path $\x_{\mathcal{S}}(t)$ is determined by:
\begin{equation}
\frac{d \log p_{\ct}(\bxS(t))}{dt} = -\Tr(\hat{\mathbf{J}}_t(\bxS(t)))
\label{eqa:cov-sub}
\end{equation}
\end{theorem}
where $\hat{\mathbf{J}}_t \in \mathbb{R}^{D_{S} \times D_{S}}$ is the sub-space Jacobian matrix of the sealed vector field $\hat{v}$.

\begin{proof}
By the chain rule, the total derivative of the sliced density $p_{\ct}$ along the path $\x_\mathcal{S}(t)$ is given by:
\begin{align*}
\frac{d \log p_{\ct}(\x_\mathcal{S}(t))}{dt} &= \frac{\partial \log p_{\ct}(\x_\mathcal{S}(t))}{\partial t} + \frac{d (\x_{\mathcal{S}}(t))}{dt} \cdot \nabla_{\bxS} (\log p_{\ct}(\x_\mathcal{S}(t)))\\ 
&= \frac{1}{p_{\ct}}\frac{\partial p_{\ct}}{\partial t} + \hat{{v}}_{\mathcal{S}}\cdot\nabla_{\bxS} (\log p_{\ct}), 
\end{align*}
where in the last passage we have simplified notation ($p_{\ct}(\bxS)=p_{\ct}$).
Substituting the restored local continuity equation~(Eq.\ref{eqa:restored}) into the above expression, we have:
\begin{equation*}
\frac{d \log p_{\ct}}{dt} = \frac{1}{p_{\ct}} \left[ - \nabla_{\bxS} \cdot (p_{\ct} \hat{{v}}_{\mathcal{S}}) \right] + \hat{{v}}_{\mathcal{S}} \cdot \nabla_{\bxS} ( \log p_{\ct})
\end{equation*}
Applying the product rule again for the substituted term:
\begin{equation}
\frac{d \log p_{\ct}}{dt} = - \left( \nabla_{\bxS} \cdot \hat{{v}}_{\mathcal{S}} \right) - \frac{1}{p_{\ct}} (\hat{{v}}_{\mathcal{S}} \cdot \nabla_{\bxS} p_{\ct})+ \hat{{v}}_{\mathcal{S}} \cdot \nabla_{\bxS} (\log p_{\ct})
\end{equation}
The last two terms cancel out and by recognizing that the divergence of the sealed field is the trace of its sub-space Jacobian, we arrive at the final equation in Eq.\ref{eqa:cov-sub}.
\end{proof}

\subsection{Sub-flow Hutchinson’s Trace Estimator}
\label{sec:hut}
The Jacobian matrix $\mathbf{J} \in \mathbb{R}^{D\times D}$ of a transformation reflects first-order differential properties of a function at a point. 
For CNF-based generative models, the instantaneous Jacobian $\mathbf{J}_t = \nabla v(\x(t),t) \in \mathbb{R}^{D\times D}$ contains rich information about the geometric deformation and probabilistic mass change along the trajectory $\x(t)$ under the Lagrangian perspective.
Having established the autonomous sub-flow identities with the sealed vector field $\hat{v}(x,t)$, we now introduce a stochastic tool that efficiently extract diagnostic signals for sub-flow targets from a global CNF.
Additional definitions are available in Appendix~\ref{sec:appendix:stochastic_tools}.

To compute the exact trace of $\mathbf{J} \in \mathbb{R}^{D\times D}$, $D$ separated differential operations are required.
This is unpractical for high-dimensional CNFs parametrized with neural networks.
The Hutchinson’s trace estimator as discussed in Eq.~\ref{eqa:hut} can efficiently estimate the trace using probing vectors $\epsilon$ sampled from Gaussian- or Rademacher-noise.

For the projected sub-flow $\x_\mathcal{S}(t)$ with respect to the sealed vector field $\hat{v}$, we are interested in the trace of the sub-flow Jacobian $\Tr(\hat{\mathbf{J}}_t(\x_{\mathcal{S}}(t)))$.
Due to the vector field's sealed dimensions that have zero partial derivatives on the diagonal elements, the trace of the sub-flow Jacobian is the same as the trace of the full-flow's Jacobian: $\Tr(\hat{\mathbf{J}}_t(\bxS(t))) =\Tr(\hat{\mathbf{J}}_t(\x(t)))$.
Applying the Hutchinson’s equation:
\begin{equation}
     \Tr(\hat{\mathbf{J}}_t(\x_{\mathcal{S}}(t))) =  \mathbb{E}_\epsilon[\epsilon^T \hat{\mathbf{J}}_t(\x(t)) \epsilon].
     \label{eqa:hut-1}
\end{equation}
Recalling the kinematic projection $\mathbf{K}_{\mathcal{S}}$ operator in Eq.~\ref{eq:projected_field}, we take the derivative on both sides obtaining:
\begin{align*}
    \nabla_x \hat{v}(x, t) &= \mathbf{K}_{\mathcal{S}} \nabla_x v(x, t) \mathbf{K}_{\mathcal{S}} \implies
    \hat{\mathbf{J}}_t(\x(t)) = \mathbf{K}_{\mathcal{S}}\mathbf{J}_t(\x(t))\mathbf{K}_{\mathcal{S}}.
\end{align*}
Substituting this identity to Eq.~\ref{eqa:hut-1}:
\begin{align}   
    \Tr(\hat{\mathbf{J}}_t(\x_\mathcal{S}(t))) &= \mathbb{E}_\epsilon[(\epsilon^T \mathbf{K}_{\mathcal{S}})\mathbf{J}_t(\x(t))(\mathbf{K}_{\mathcal{S}} \epsilon)] 
    =  \mathbb{E}_{\hat{\epsilon}}[\hat{\epsilon}^T \mathbf{J}_t(\x(t)) \hat{\epsilon}] 
    \approx \frac{1}{N} \sum_{n=1}^{N}\hat{\epsilon}_n^T \mathbf{J}_t(\x(t)) \hat{\epsilon}_n,
    \label{eq:trace_of_subflow_jacobian}
\end{align}
where $\hat{\epsilon} = \mathbf{K}_{\mathcal{S}} \cdot\epsilon$ is the projected probing vector.
The proposed sub-flow Hutchinson’s Trace Estimator allows estimating the sub-flow Jacobian trace from the Jacobian of the ambient space with only  a few ($N < 10$) Jacobian-vector-products~(JVPs).
This makes the sub-flow diagnostic extremely efficient for high-dimensional neural vector fields\footnote{VJP is often supported by automatic differential systems in modern machine learning frameworks e.g., Pytorch}.

For a sampled sub-trajectory $\x_\mathcal{S}(0) \to \x_\mathcal{S}(1)$, the trace of the instantaneous Jacobian $\mathbf{J}_t(\x(t))$ represents a local contraction~($\Tr < 0$) or expansion~($\Tr >0$) of probability mass along the path.
Alternatively, the sub-flow Hutchinson’s trace estimator can be used to evaluate the density $p_{\ct}(x_\mathcal{S})$ for any sample $x_\mathcal{S}$ in the sub-space directly using Eq.~\ref{eqa:cov-sub}.
An illustration of this is given in Section~\ref{sec:experiments:toy-example}.

\subsection{Metrics for OOD}
\label{sec:metrics_for_ood}
The first metric $\text{LL}=p_1(\bxS(1))$ is the log likelihood of the target sub-representation.
This can be estimated with Eq~\ref{eqa:loglikelihood} where we substitute the full Jacobian with the sub-flow Jacobian from Eq.~\ref{eq:trace_of_subflow_jacobian}, and $p_0(x(0))$ with $p_0(\bxS(0))$.
The metric $\text{PRIOR}=p_0(\bxS(0))$ estimates the sample fit to the distribution in latent space  obtained by whitening the target sub-representation.
For high-dimensional data, the probability mass tends to mainly concentrate on a sphere of radius $\sqrt{D_\mathcal{S}}$.
To correct for this phenomenon, the metric $\text{RAD}=-\frac{1}{D_\mathcal{S}}(||\bxS(0)||^2-D_\mathcal{S})^2$ compares the L2 norm of $\bxS(0)$ to the expected radius.

Because CNFs trained with flow-matching optimal-transport paths are encouraged to generate straight paths for ID samples, we also introduce two heuristic metrics that estimate deviations from this behaviour.
DISP computes the ratio between the displacement of the OT path $|\bxS(1)-\bxS(0)|$ and the whitening path $\int_0^1|v_k(\bxS(t), t)|dt$.
This should be close to 1 for straight paths and lower for OOD samples.
Finally, COS integrates the cosine similarity between the OT path and the whitening path:
\begin{align}
    \text{COS} = \int_0^1 \cos\left[\bxS(t)-\bxS(0), \bxS(1)-\bxS(0)\right]dt.
\end{align}

All the above metrics suffer from the un-normalized nature of the sub-flow probability density discussed in Section~\ref{sec:coupled_sub-flow}.
A commonly adopted solution is to resort to relative measures computed with different models (a dedicated and a background model) or in different conditions (for example using different contextual information).
In the case of LL, for example, this corresponds to the commonly used likelihood ratio between the two conditions.
In Section~\ref{sec:results} we compare the performance of absolute and relative metrics.
The specific definitions of the metrics for our experiments are given in Appendix~\ref{sec:appendix:pronunciation}.

\section{Experiments}
\label{sec:experiments}

\begin{figure}
    \centering
    \includegraphics[width=1\linewidth]{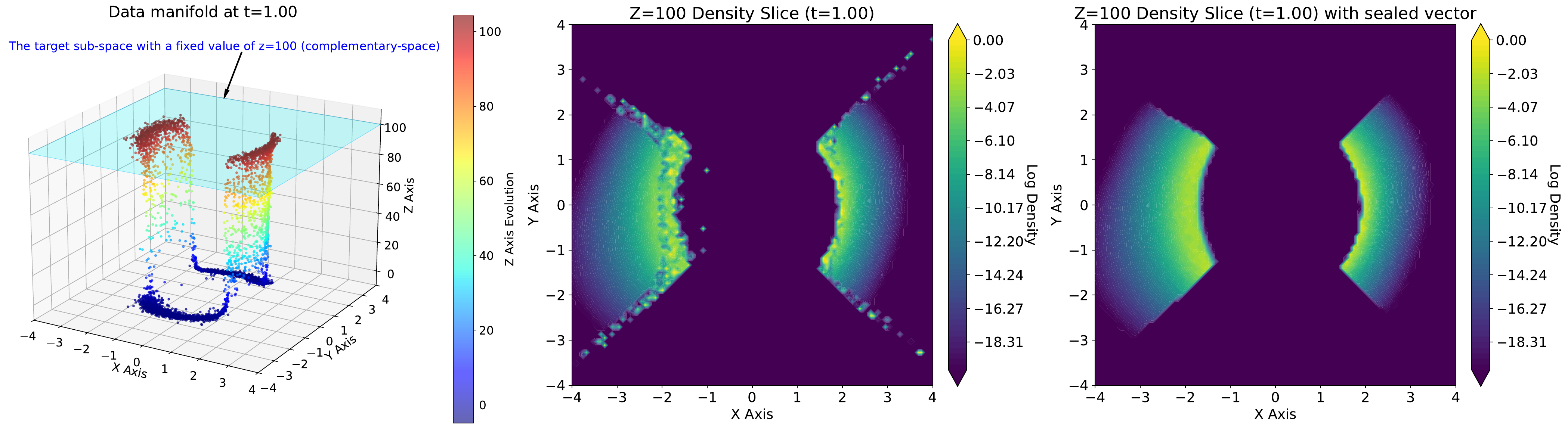}
    \caption{Illustration of density estimation with and without the sealing mechanism.}
    \label{fig:toy-example}
\end{figure}

\subsection{Synthetic Data Illustration}
\label{sec:experiments:toy-example}
In order to exemplify our method, we first use synthetic data from a simple distribution.
The data space is $\mathbb{R}^3$, where the target subspace $\mathcal{S}$ is the $xy$-plane, and the complementary subspace $\mathcal{C}$ is the $z$-axis.
In the $xy$-plane, we train a 2D vector field $v_\theta(x,y)$, using flow matching with optimal transport~(OT) paths from Gaussian noise to the circle $x^2+y^2=4$.
We then force an additional velocity in the $z$-axis according to $v_z= (x^2-y^2)*(z+\textit{const})$ with limit at $v_z=100$.
Fig.~\ref{fig:toy-example}~(left) shows samples from the corresponding data distribution.
If we consider a subspace $\mathcal{S}$ for $z=100$ ($\mathbf{c}(1)$), we can observe that the data lies on a manifold of $\mathcal{S}$.
However, if we try to estimate the density over $\mathcal{S}$ using the traditional trace estimator, additional divergence is introduced by the third dimension due to the compression/expansion of the velocity $v_z$ along its axis.
Consequently, the densities of many points on $\mathcal{S}$ are polluted (Fig.~\ref{fig:toy-example}, center) and the likelihood for out-of-distribution points is often overestimated.
On the other hand, the density on a manifold appears to be smooth and complete with the help of the proposed approach (Fig.~\ref{fig:toy-example}, right). 

\subsection{Real Data: Zero-Shot Mispronunciation Detection with CNF-based Generative Models}
\label{sec:experiments:pronunciation}
Consider a spoken utterance $x_1$ and a transcription \Lcano that contains the sequence of $K$ phonemes in the expected canonical pronunciation of $x_1$.
Phoneme-level Mispronunciation Detection and Diagnostic~(MDD) aims to assess the quality of the realizations for each phoneme in \Lcano, given $x_1$.
This task can be thought of as an instance of OOD detection if we have a probabilistic model of the correct pronunciations.
The spoken utterance is usually represented by a Mel-spectrogram, that is, $x_1 \in \mathbb{R}^{L \times M}$ where $L$ is the number of spectral feature vectors and $M$ is their dimensionality.
A source of complexity for the problem lies in the fact that both $L$ and $K$ are variable for each utterance, and that there is no one-to-one correspondence between the feature vectors and the corresponding phonemic targets.
This problem has been traditionally approached using models for automatic speech recognition (ASR).
Here we propose to use generative models for text-to-speech (TTS) synthesis, instead.
Details on CNF-based models for TTS and the motivation are given in Appendix~\ref{sec:appendix:TTS} and \ref{sec:appendix:pronunciation}.
Here we focus on the way we use the proposed sub-flow method to implement phoneme-level MDD.

The first step in the method is to segment the test utterance using forced alignment driven by an ASR model.
This assigns each feature vector in $x_1$ to a specific phoneme $k$.
Then we define the target sub-space $\mathcal{S}$ as the sequence of feature vectors corresponding to the target phoneme $k$ and the complementary space for all other feature vectors.
This allows us to consider as context any characteristics of the speaker voice, recording channel, and environment that are not related to the correctness of pronunciation.
Finally, we compute the goodness of pronunciation (GOP) using the metrics in Section~\ref{sec:metrics_for_ood} (implementation details in Appendix~\ref{sec:appendix:pronunciation}), and use them to separate correct from incorrect pronunciations.

\subsubsection{Experimental Settings}
We evaluate our methods on the CMU Kids dataset~\cite{CMU_kids} which is a publicly available dataset comprising 9.1 hours of speech~(5,180 utterances) collected from 78 children~(24 male, 54 female) aged 6–11 in U.S. schools. 
The corpus is uniquely valuable due to its phonetic annotations of the erroneous utterances.
By realigning the canonical pronunciations derived from the CMU pronunciation dictionary to the human-labeled phonetic annotations, we obtain phoneme-level labels~(correct/incorrect) for the full corpus.
As a performance metric we report ROC-AUC (area under the curve) and the corresponding 95\% confidence intervals.

We use a publicly available pretrained model for TTS called F5-TTS~\cite{chen-etal-2024-f5tts}.
This is a CNF-based model trained with flow-matching with optimal-transport path on roughly 100,000 hours of diverse speech data (details in Appendix~\ref{sec:appendix:TTS}).  
We use the most active checkpoint~(F5TTS\_v1\_Base/1250000) to generate the flow.
Segmentation of the utterances into phonetic segments is obtained by a mono-phone forced-aligner that we trained using Kaldi~\cite{Povey_ASRU2011}.
We use Euler's method as the ODE-Solver implemented by \cite{odenet} with 32 integral steps.
F5-TTS uses a non-uniform sampling strategy~(sway factor SF) to boost the sampling quality. 
We use the default value of SF=-1.
We also disable the classifier-free guidance~(CFG) used to increase the sample fidelity that is not relevant for MDD. 
Experiments were conducted on a workstation equipped with an AMD Ryzen Threadripper 3960X CPU, 126 GB of physical RAM, and two NVIDIA GeForce RTX 3090 GPUs (24 GB VRAM each).
The software environment was based on a standard GNU/Linux x86\_64 distribution.

\begin{table}
\centering
\caption{Comparison of absolute and relative measures.}
\begin{tabular}{@{}lcccc@{}}
\toprule
\textbf{Method}      & \multicolumn{2}{c}{\textbf{Absolute Measures}} & \multicolumn{2}{c}{\textbf{Relative Measures}} \\ \midrule
                     & \textbf{AUC} & \textbf{95\% Conf} & \textbf{AUC} & \textbf{95\% Conf} \\ \midrule
GOP-LL     & $0.641$   & $5.5\times 10^{-3}$ & $0.668$ & $5.3\times 10^{-3}$ \\
GOP-PRIOR  & $0.545$   & $6.3\times 10^{-3}$ & $0.596$ & $5.9\times 10^{-3}$ \\
GOP-RAD    & $0.578$   & $6.0\times 10^{-3}$ & $0.391$ & $6.6\times 10^{-3}$ \\
GOP-DISP   & $0.525$   & $6.3\times 10^{-3}$ & $0.701$ & $5.0\times 10^{-3}$ \\
GOP-COS    & $0.516$   & $6.3\times 10^{-3}$ & $0.738$ & $4.6\times 10^{-3}$ \\
\midrule
GOP-Codec\cite{valle-cao}          & $0.492$ & $6.4\times 10^{-3}$ & $0.765$ & $4.2\times 10^{-3}$ \\
GOP-GMM\cite{cao23_interspeech}    & -     & -      & $0.723$ & $4.7\times 10^{-3}$ \\
GOP-CTC-SF\cite{cao_3}             & -     & -      & $0.915$ & $2.0\times 10^{-3}$ \\ 
\bottomrule
\end{tabular}

\label{tab:gop_metrics}
\end{table}
\subsection{Results and Discussion}
\label{sec:results}
Table~\ref{tab:gop_metrics} reports the performance of the different metrics for the mispronunciation detection task and both for absolute and relative measurements (see Section~\ref{sec:metrics_for_ood}).
GOP-LL outperforms the GOP-PRIOR by a large margin, which indicates that considering solely the end point $\bxS(0)$ of the whitening process for MDD is not enough. 
This means that, while all the speech segments (correct/incorrect) can be roughly pulled back by the vector field to the isotropic Gaussian noise~(shown by AUC$\approx$0.5), the additional local compression rate of probability mass along with the whitening trajectory $\bxS(t)$ recorded by the Jacobian trace in Eq.~\ref{eqa:cov-sub} contains informative diagnostic signals for the tasks. 
The absolute measure of GOP-RAD seems to mitigate the concentration of mass in the high-dimensional space if compared to GOP-PRIOR.
However, this result does not carry over to the relative versions of the metrics and relative GOP-RAD appears to be inversely correlated with the goodness of pronunciation (AUC$<0.5$).

The likelihood-based approaches (GOP-LL, GOP-PRIOR, GOP-RAD) under-perform if compared to the geometric methods (GOP-DISP, GOP-COS) and especially when comparing to the baseline TTS-based approach GOP-Codec and ASR-based approaches GOP-GMM and GOP-CTC-SF.   
We hypothesize that this is due to but not limited to the following reason.
The sub-flow is defined in the high dimensional space $\mathbb{R}^{L_k \times M}$ while the human speech conditioned on its speech context $C_{sp}$ live in a lower-dimensional manifold. 
As a result, in the Lagrangian perspective, most of the model's energy (e.g., the compressions represented by the Jacobian traces) is consumed to push the noise onto the manifold in order to maintains inter-frame coherence along the trajectory. 
It is hard for the likelihood-based method to extract those fine-grained signal that determines phonetic/phonemic correctness, even though the LSF framework has the effect of purifying the diagnostic signal. 
This hypothesis is consistent with many previous works.
For example, \citeauthor{nalisnick2019do}~\cite{nalisnick2019do} discovered that generative models like flow-based models and VAEs often assign a higher likelihood to simpler OOD samples than to its own training data.
\citeauthor{xiao2020likelihood}~\cite{xiao2020likelihood} suggested to use likelihood-ratio test to tackle this issue, which shares very similar idea to the GOP and to our relative measures which achieves constantly improvement on the AUC throughout the experiments.    
\citeauthor{kirichenko2020why}~\cite{kirichenko2020why} further attributed this phenomenon to the ``inductive bias" of normalizing flows towards local pixel correlations rather than modeling high-level semantic concepts.
In the work of \cite{lelan2021geometric}, the authors proposed to estimate local intrinsic dimensions for Flow-based models and use it as an additional criterion to the likelihood for OOD detection.

\begin{figure}
    \centering
    \includegraphics[width=1\linewidth]{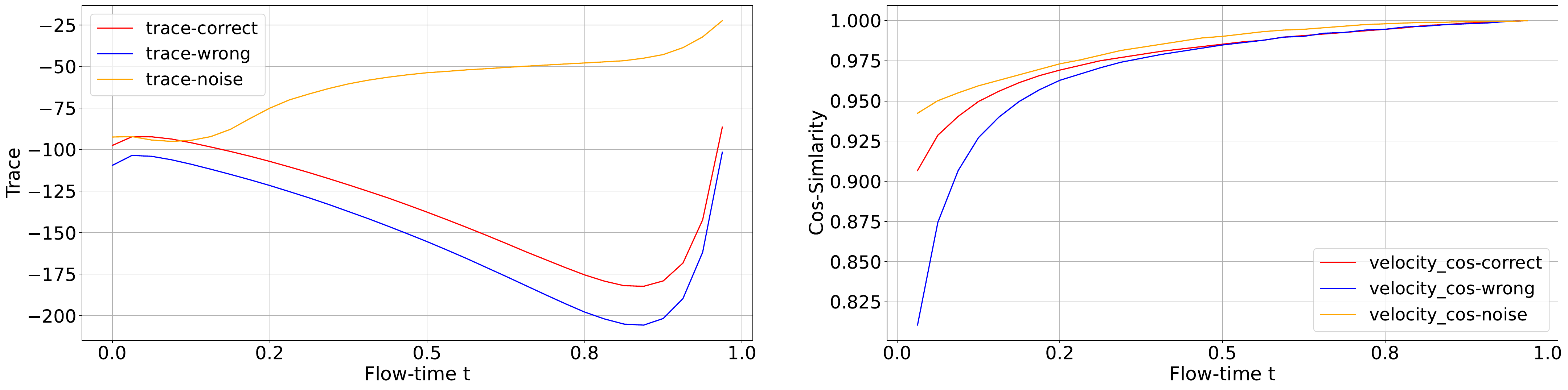}
    \caption{The diagnostic signals extracted from the F5-TTS-based phonetic sub-flow for the target phoneme /AE/ in the utterance for ``cat''.}
    \label{fig:signals}
\end{figure}
The simple geometric-based measures GOP-DISP and GOP-COS appear to be more performantive than likelihood-based methods.
The relative measure of GOP-COS surpasses the traditional GMM-based GOP method with AUC=0.738.
A way to explain this is to consider the example in Fig.~\ref{fig:signals}.
The figure shows the Jacobian trace~(left) and cosine similarity~(right) of the target sub-flow as a function over flow-time $t$ for three samples (correct pronunciation, mis-pronunciation and noise).
This is an example of the ``likelihood paradox'' where the trace for the mis-pronunciation results in a higher likelihood than for the correct pronunciation.
Only the noise sample can be properly distinguished from correct pronunciation in this case.
Conversely, the cos-similarity plot shows a better adherence to the OT path for the correct pronunciation than for the mis-pronunciation, as expected.
  
\subsection{Limitations}
\label{sec:limitations}
This work is a first attempt at using CNF for sub-space OOD detection. 
Our experiments are limited to the specific real-world application of mispronunciation detection.
Although this is a challenging testbed, other application domains should be tested in order to confirm the validity of the method.
For the likelihood methods, the complexity is dictated by the number of samples used in the Jacobian estimation.
For the non-likelihood methods, the predominant factor is the number of evaluations along the flow path.

\vspace{-2mm}
\section{Conclusions}
In this work, we propose a framework for using generative models based on continuous normalized flow (CNF) for out-of-distribution detection tasks.
The method is intended for high-dimensional data, where the detection task is based on a subspace and the rest of the data representation is used as context.
We use concepts from fluid dynamic to enable the definition of diagnostic signals based on the sub-flow corresponding to the target subspace.
Additionally, we provide a numerical tool for performing tractable trace estimation for the sub-flow, allowing for efficient computations.

We test our methods on the challenging task of zero-shot phoneme-based pronunciation assessment and using CNF-based generative models for text-to-speech synthesis.
We show that our method has a comparable performance to traditional methods based on discriminative models for automatic speech recognition (GOP-GMM), although they come short when compared to SOTA discriminative methods (GOP-CTC-SF).
Furthermore, we show that simple geometric analysis of the whitening trajectory surpasses likelihood-based diagnostics in our tests.
This confirms that geometric mismatch is a more sensitive indicator for OOD detection than simple density fluctuations with CNF-based models.

Future work should verify if our results generalize to other application domains, for example in computer vision or sensor systems.

\printbibliography

\clearpage
\appendix

\begin{figure}
    \centering
    \includegraphics[width=0.8\linewidth]{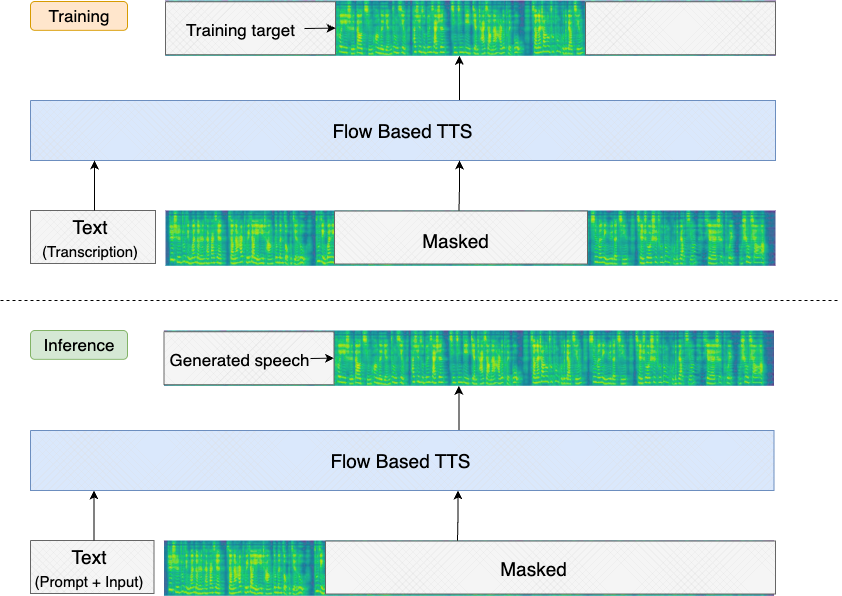}
    \caption{The masking and filling paradigm for flow-based TTS}
    \label{fig:mnf}
\end{figure}

\section{Speech Synthesis Using Continuous Normalizing Flows}
\label{sec:appendix:TTS}
Text-to-speech (TTS) systems use generative speech models for synthesizing human-like speech corresponding to given text inputs~(TI), striving for both fidelity and diversity of the output.
Recently, TTS systems have seen significant improvements thanks to the advances in deep neural generative models.
Among them, the models that are based on CNF~\cite{voicebox,E2,chen-etal-2024-f5tts} have achieved the state-of-the-art performance on various benchmarks.   

TTS systems operate on variable length speech utterances that are typically represented as sequences of feature vectors.
In particular, CNF-based TTS systems usually employ a mel-spectrogram representation where the  dimensionality $M$ of the feature vectors (mel-spectral bins) is fixed, whereas the number of vectors $L$ can vary for each utterance.
The total dimensionality of each observation in the data is then $D=L\times M$.
The flow of TTS is characterized by a time-dependent velocity field that is parametrized by stack of Non-Auto-Regressive~(NAR) transformer layers~\cite{attn, chang2022maskgit}. 
Similarly to the ordinary CNF, a time-dependent 2-D tensor $\mathbf{x}(t)$ represents a trajectory within the flow that transports a sample of Gaussian noise $\x(0)$ to the speech data $\x(1)$.

These models are capable of generating target speech for unseen speakers~(zero-shot TTS) with in-context learning.
This is achieved by providing the neural ODE with additional conditions, i.e., a matched pair of a short speech prompt~(SP) from the target speaker and the corresponding orthographic transcription (text prompt, TP, see Figure~\ref{fig:mnf}).
The system, therefore, has three inputs, a noise sample $x_0$, a speech prompt SP and a concatenation of the text prompt and the text input $\{\text{TP}, \text{TI}\}$.
Both the speech prompt and the text are converted to the same dimensionality $L\times M$, by padding and mapping, resulting in the context matrices $C_\text{sp}$ and $C_\text{text}$.
Finally, the three inputs $x_0, C_\text{sp}$, and $C_\text{text}$ are first stacked and then mapped to the final latent representation by linear mapping.
The network then estimates the velocity field $v_\theta(\x(t), t; C_{\text{sp}},C_{\text{text}})$ for the trajectory $\x(t)$ in the context $C_{\text{sp}},C_{\text{text}}$ that can be plugged in the ODE (Eq.\ref{eqa:vf}).

Under these settings, the training of the model becomes very similar to training a masked language model~(MLM)~\cite{bert} where the model is learned to recover the masked part of speech driven by the conditional vector field $v_\theta(x,t;C)$.
The training and inference modes of CNF-based TTS are illustrated in Fig.\ref{fig:mnf}.

\begin{figure}
    \centering
    \includegraphics[width=\linewidth]{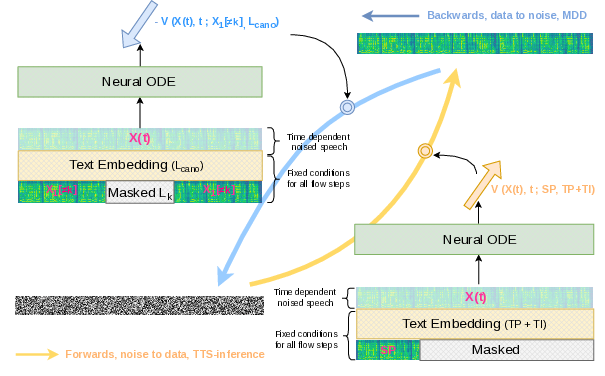}
    \caption{The whitening trajectory for MDD~(backwards) vs. the inference trajectory for TTS~(forwards) of a CNF-based speech generative model.}
    \label{fig:flow}
\end{figure}

\section{Details on Mispronunciation Detection with CNF-based Generative Models}
\label{sec:appendix:pronunciation}
Here we give details on the proposed application of our method to mispronunciation detection that were not included in the paper.

Given a spoken utterance $x_1$ and a transcription \Lcano that contains the sequence of $K$ phonemes in the expected canonical pronunciation, phoneme-level MDD aims to assess the quality of the realizations for each phoneme in \Lcano, given $x_1$.
The spoken utterance is usually represented by a Mel-spectrogram, that is, $x_1 \in \mathbb{R}^{L \times M}$ where $L$ is the number of spectral feature vectors and $M$ is their dimensionality.
The complexity of the problem lies in the fact that both $L$ and $K$ are variable for each utterance, and that there is no one-to-one correspondence between the feature vectors and the corresponding phonemic targets.

To leverage a well-trained CNF-based speech generative model and re-purpose it for phoneme-level MDD, we propose temporal-aware phonetic sub-flow diagnostic~(TPD) based on the reversibility of CNF.
This can be performed in three steps:
\begin{enumerate}
    \item Temporal segmentation: We obtain $K$ ordered segments $x_1 = \{x_1[1],...,x_1[K]\}$ where $x_1[k] \in \mathbb{R}^{L_k \times D}$ denotes the ``$k$''th segment of length $L_k$, using techniques such as monotonic-alignment search or forced alignment.
    \item Whitening: The assessment speech $x_1$ is inversely filtered by running the CNF backwards in time from $t=1 \rightarrow 0$ driven by the conditional vector field $v_k(x,t;C_k)$.
    The conditions $C_k = \{C_{\text{text}} \leftarrow L_{\text{cano}}, C_{\text{sp}} \leftarrow x_1[ \neq k]\}$ are rendered for the ``$k$''th phoneme where $x_1[ \neq k]$ denotes the original speech $x_1$ whose segment $x_1[k]$ is replaced by special mask tokens of the same length.
    \item Projecting the sub-flow: We focus on the local whitening trajectory $\x_k(t)$ and analyze the sub-flow dynamics for MDD using local kinematic diagnostic signals.
\end{enumerate}
It should be noted that the masking applied to the target speech in the second step is critical because we should regard only $x_1[ \neq k]$ as necessary context such that the semantic of masking is consistent with that for generation~(Fig.\ref{fig:mnf}).  
In case of no mask, the distribution of the target segment peaks at the point identical to the provided speech segment. 
The neural ODE is then at the risk of receiving too strong signals from target segment while omitting other conditions for differentiating between good and bad pronunciations.
We illustrate the proposed TPS in Fig.\ref{fig:flow}.

The advantages of TPD can be summarized as follows: 
1. It is a zero-shot MDD approach with a non-discriminative model.
Similarly to ASR-based solutions, this mitigates the data-scarcity problem typical of MDD tasks.
2. We preserve the global context awareness during whitening, while keeping an eye on the local dynamics. 
3. TPD takes advantage of the probabilistic models while at the same time begin totally deterministic via whitening, so no information of the model is lost as ``sampling-comparing" based approaches.
4. Since the analysis of the sub-flow is performed at the latent speech space, no vocoder is needed to reconstruct the waveform which not only simplifies the workflow but also avoids the additional variance brought by phase restoration. 
5. The masking mechanism enables flexible speech assessment for any linguistic-level, by just expanding or reducing the the span of the mask.

\subsection{Mispronunciation Detection Using the Sub-Flow}
In the next section, we specialize the kinematic measures introduced in Section~\ref{sec:metrics_for_ood} for the MDD task.

The phoneme-level MDD can be framed as out-of-distribution~(OOD) detection tasks for each phonemic unit $l_k \in L_{\text{cano}}$, using the corresponding sub-flow trajectory $\x_k(t)$ driven by the sealed vector field $\hat{v}_k(x,t;C_k)$ where $C_k = \{C_{\text{text}} \leftarrow L_{\text{cano}}, C_{\text{sp}} \leftarrow x_1[ \neq k]]\}$.
The first measure comes naturally to our mind is the log-likelihood value evaluated for each segment $x_1[l]$ using the instantaneous change of variables for sub-flows according to Eq.\ref{eqa:cov-sub}:
\begin{equation}
     \text{GOP-LL}(k)=\frac{1}{L_k}\left(\log p_0(\x_k(0)) + \int_0^1-\Tr(\hat{\mathbf{J}}_t(\mathbf{x}_k(t)))dt\right),
\end{equation}
where the sub-flow Hutchinson’s trace estimator introduced in Sec.\ref{sec:hut} can be used to approximate the instantaneous sub-flow Jacobian trace.
Intuitively, the in-domain~(ID) samples should have a higher GOP-LL than OOD samples.

According to the Neyman-Pearson Lemma, the likelihood-ratio~(LR) test is the most powerful test given a fixed significance level.
For this reason, we introduce a second trajectory $\bar{\x}_k(t)$ driven by a relative sealed vector field $\hat{v}_k(x,t;\bar{C_k})$ where $\bar{C}_k = \{C_{\text{text}} \leftarrow L_{\text{NULL}}, C_{\text{sp}} \leftarrow x_1[ \neq k]]\}$ and denote its Jacobian as $\bar{\mathbf{J}}$.
The uninformative text condition $L_{\text{NULL}}$ instructs the neural ODE to generate an average vector field for the given speech context.
The corresponding log-likelihood score is:
\begin{equation}
     \text{GOP-LL-NULL}(k)=\frac{1}{L_k}\left(\log p_0(\bar\x_k(0)) + \int_0^1-\Tr(\bar{\mathbf{J}}_t(\bar\x_k(t)))dt\right).
\end{equation}
GOP-LL-NULL has the effect of normalizing out non-phonemic factors of the realization of the target phoneme when subtracted by the ordinary log-likelihood score.
We define this relative log-likelihood measure as:
\begin{equation}
    \text{GOP-LL-R}(k) =   \text{GOP-LL-NULL}(k) - \text{GOP-LL}(k). 
\end{equation} 
We omit the definitions of this relative scores for the rest of the measures, but keep in in mind that this idea of using relative score is applied throughout all the experiment.

Since the sub-space trajectory $\x_k(t)$ is obtained using whitening, a simpler idea than GOP-LL is to evaluate the initial probability $p_0(\x_k(0))$ as measure for MDD:
\begin{equation}
    \text{GOP-PRIOR}(k) =  \frac{1}{L_k}(\log p_0(\x_k(0)).
\end{equation}
Due to the effect of concentration of mass for high-dimensional data, the probability mass is, in theory, concentrated mainly on a sphere of radius $\sqrt{L_k M}$.
An alternative to GOP-PRIOR for mitigating this problem is then computing the L2 norm of $\x_k(0)$ and comparing to the expected radius:
\begin{equation}
    \text{GOP-RAD}(k) =  -\frac{1}{L_k}(||\x_k(0)||^2-L_k M)^2.
\end{equation}

Since the CNF trained with Flow-Matching on Optimal-Transport path will in principle generate constant-velocity straight trajectories for ID samples, as illustrated Eq.\ref{eqa:OT-path}, we propose several efficient heuristic measures based on simple geometric properties derived from the whitening path $\x_k(t)$ and the corresponding projected vector field $v_k(\x_k,t;C_k)$.

In transport theory, OT paths represent the minimum kinetic energy path for an autonomous dynamical system.
The first of the proposed kinematic measures computes the ratio between displacement of the OT path and the whitening path:
\begin{equation}
    \text{GOP-DISP}(k) = \frac{|\x_k(1)-\x_k(0)|}{\int_0^1|v_k(\x_k(t),t;C_k)|dt}.
\end{equation}
We expect GOP-DISP to be close to one for ID samples while less than one for OOD samples.
As directions might be a more direct implication for straightness, the cos-similarity can also be a normalized measure for detecting OOD samples:
\begin{equation}
    \text{GOP-COS}(k) =  \int_0^1\textbf{cos}\left(\x_k(t)-\x_k(0), \x_k(1)-\x_k(0)\right)dt.
\end{equation}

\section{Additional Stochastic Tools for Diagnostics}
\label{sec:appendix:stochastic_tools}

\subsection{Sub-flow Randomized Directional Derivatives Analysis}
While the sub-flow trace estimator measures the aggregated density compression, it is insensitive to local isochoric distortions, e.g., shear or rotation, where the mass is rearranged but not significantly compressed.
To capture those signals which are potentially relevant for OOD detection and to keep computational efficiency, we propose the sub-flow randomized directional derivatives~(RDD) that analyze more fine-grained geometric deformation of the sub-flow.

Rather than explicitly evaluating the full Jacobian, we probe the sub-flow Jacobian operator and collect the directional responses using N random directions $\{{\epsilon }^{(i)}\}_{i=1}^{N})$:
\begin{equation}
    \hat{\mathbf{J}}_t(\x(t)) \epsilon^{(i)} = \mathbf{K}_{\mathcal{S}}\mathbf{J}_t(\x(t))(\mathbf{K}_{\mathcal{S}} \epsilon^{(i)}) = \mathbf{K}_{\mathcal{S}}{\mathbf{J}}_t(\x(t)) \hat\epsilon^{(i)}.
\end{equation}
Similar to the Sub-flow Hutchinson’s Trace Estimator, the last equation exposes a double-sealing mechanism: the right $\mathbf{K}_{\mathcal{S}}$ aims to ensure the variations in the random tangent vector $\epsilon$ only originate from the targeting sub-space while the left $\mathbf{K}_{\mathcal{S}}$ masks out the off-sub-space dimensions in the Jacobian responses. 
According to the concentration of measure phenomenon in high-dimensional spaces, a small number of random probes are sufficient to capture the essential spectral properties~(principle directions) of the local Jacobian.
These directional responses preserve the orientational information of the local flow.
To detect ODD, we further propose two diagnostic signals based on the orientational information: 
\subsubsection{Temporal Orientational Consistency}
This metric captures the second-order Jacobian stability by computing the average cosine similarity between successive directional responses using the same probe vectors along the streamline path $\x(t)$:
\begin{equation}\beta(t) = \frac{1}{N}\sum_{n \in [1,N]} \cos\left( \hat{\mathbf{J}}_t(\x(t)) \epsilon^{(n)}, \hat{\mathbf{J}}_{t-\Delta t}(\x(t-\Delta t)) \epsilon^{(n)} \right)
\label{beta}
\end{equation}
where $\Delta t$ is the interval between two discrete ODE evaluation steps. 
We hypothesis that the velocity field evolves according to a pre-defined prior for ID samples.
For example, if a CNF is trained with OT path, the Jacobian orientations should remain stable along the trajectory, leading to \(\mathcal{S}_t \approx 1\). 
In contrast, an OOD sample can have rapid re-orientation of the Jacobian on its path such that significant drop in \(\mathcal{S}_{t}\) can be observed, providing a high-sensitivity diagnostic signal of the sub-flow’s intrinsic struggle against the learned prior. 

\subsubsection{Intra-sub-space Spatial Coherence}
Beyond temporal stability, we further investigate the instantaneous coherence across the constituent components of a sub-flow. 
Given that high-dimensional data in CNFs are often represented as organized structures~(e.g., sequences of frames in speech signals or spatial grids in images), we define a metric to evaluate the local alignment of directional responses within these sub-dimensions at each flow time $t$.
Formally, for a sub-flow partitioned into $M$ local units, we measure the spatial coherence of the sub-flow over each component $k$ in the sub-space Jacobian responses denoted as $\hat{\mathbf{J}}^{(k)}_t(\x(t)) \epsilon^{(i)}$: 
\begin{equation}
    \alpha(t) = \frac{1}{N}\sum_{n=1}^{N}
    \left[
    \frac{1}{K-1}\sum_{m=1}^{K-1}\cos\left( \hat{\mathbf{J}}^{(k+1)}_t(\x(t)) \epsilon^{(n)}, \hat{\mathbf{J}}^{(k)}_t(\x(t)) \epsilon^{(n)} \right)
    \right].
    \label{eqa:alpha}
\end{equation}
Our hypothesis is, in a well-conditioned sub-flow, the latent geometry should exhibit harmonic stretching, where local units within the same semantic sub-flow share a consistent deformation trend. 
Conversely, an OOD sample is characterized by internal stochastic divergence across these units. 
This diagnostic signal provides a complementary lens to temporal analysis.

\begin{figure}
    \centering
    \includegraphics[width=1\linewidth]{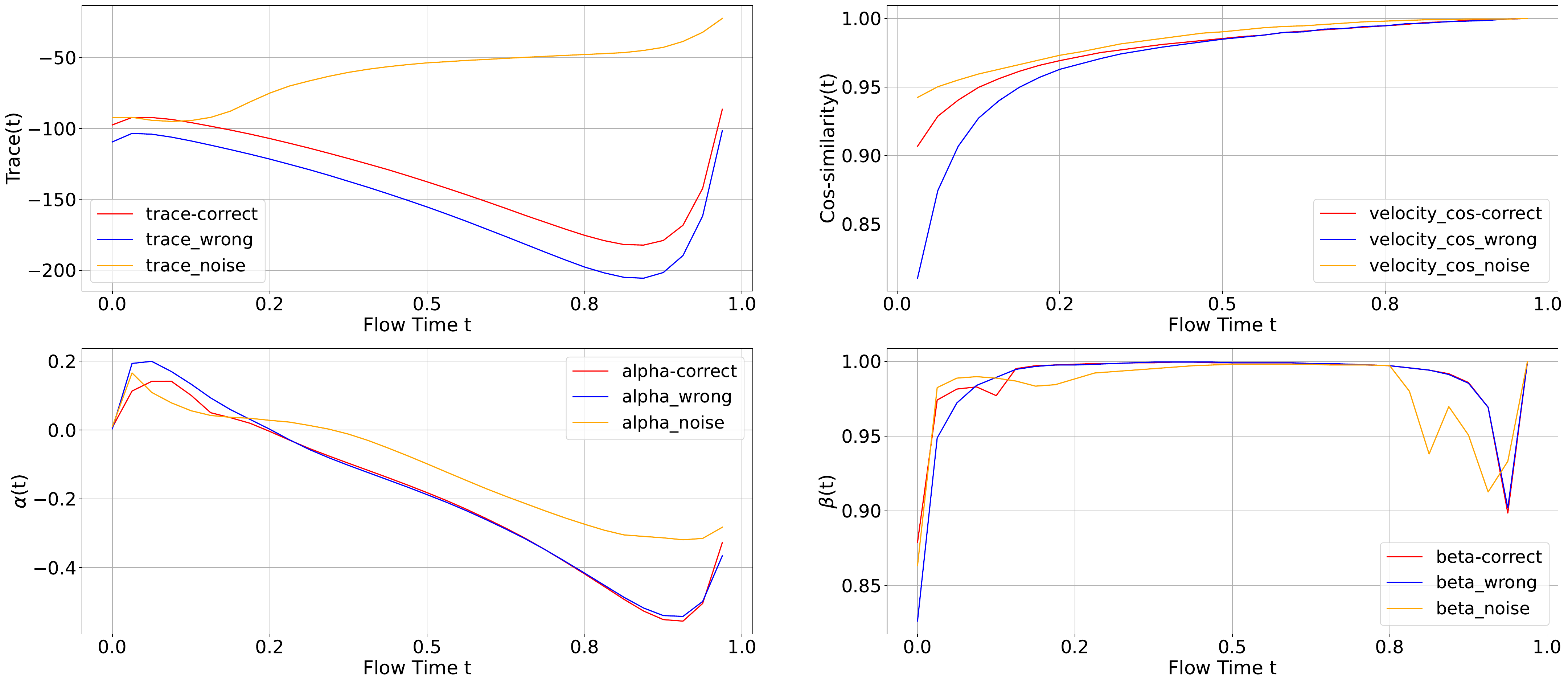}
    \caption{The diagnostic signals extracted from the F5-TTS-based phonetic sub-flow for the target phoneme /AE/ in the utterance for ``cat''.}
    \label{fig:all_signals}
\end{figure}

\subsection{Illustration of Diagnostic Signals}
As an example Figure~\ref{fig:all_signals} shows all the diagnostic signals evaluated on the same examples as in Figure~\ref{fig:signals}.
The illustration demonstrates how these signals provide a rich representation of the flow trajectories that contain far more information than the global averages.
Analyzing these trajectories in detail may potentially lead to improved performance for OOD tasks.
For example, the negative $\alpha(t)$ values in the middle stage of the trajectories show that the neighboring frames are pulled against each other, which is consistent with the previous hypotheses that generative models are modeling background or inter-pixel/frame coherence.
Apart from that, we observed how the initial part of the trajectory often deviates significantly from the rest of the trajectory.
Anecdotal evidence suggests that defining metrics only considering this part of the trajectory may be more informative for ODD detection.
This, however requires further experimental validation.

\end{document}